\documentclass[a4paper]{article}

\usepackage{amsmath}
\usepackage{amssymb}
\usepackage{graphicx}
\usepackage[OT4]{fontenc}
\usepackage[latin2]{inputenc}
\usepackage{url}
\usepackage{setspace}

\newcommand{\okra}[1]{\left( #1 \right)}
\newcommand{\kwad}[1]{\left[ #1 \right]}
\newcommand{\klam}[1]{\left\{ #1 \right\}}

\newcommand{\sred}[1]{\mathbf{E}\, #1}

\newcommand{\boole}[1]{{\bf 1}{\klam{#1}}}
\newcommand{\peq}{\stackrel{+}{=}}
\newcommand{\pge}{\stackrel{+}{>}}
\newcommand{\ple}{\stackrel{+}{<}}
\DeclareMathOperator{\card}{card}
\DeclareMathOperator{\PPM}{PPM}
\DeclareMathOperator*{\hilberg}{hilb}

\newtheorem{definition}{Definition}

\newtheorem{theorem}{Theorem}

\newenvironment*{proof}{\begin{trivlist}\item[]
\noindent\textbf{Proof:}}{$\Box$\par\end{trivlist}}
\newenvironment*{proof*}[1]{\begin{trivlist}\item[]
\noindent\textbf{Proof of #1:}}{$\Box$\par\end{trivlist}}

\begin{document}

\pagestyle{empty}   
\begin{titlepage}

  \title{Is Natural Language a Perigraphic Process?\\
    The Theorem about Facts and Words Revisited}

  \author{{\L}ukasz D\k{e}bowski\thanks{
    {\L}. D\k{e}bowski is with
    the Institute of Computer Science, Polish Academy of Sciences, 
    ul. Jana Kazimierza 5, 01-248 Warszawa, Poland 
    (e-mail: ldebowsk@ipipan.waw.pl).}}
\date{}

\maketitle

\begin{abstract}
  As we discuss, a stationary stochastic process is nonergodic when a
  random persistent topic can be detected in the infinite random text
  sampled from the process, whereas we call the process strongly
  nonergodic when an infinite sequence of independent random bits,
  called probabilistic facts, is needed to describe this topic
  completely. Replacing probabilistic facts with an algorithmically
  random sequence of bits, called algorithmic facts, we adapt this
  property back to ergodic processes. Subsequently, we call a process
  perigraphic if the number of algorithmic facts which can be inferred
  from a finite text sampled from the process grows like a power of
  the text length. We present a simple example of such a
  process. Moreover, we demonstrate an assertion which we call the
  theorem about facts and words. This proposition states that the
  number of probabilistic or algorithmic facts which can be inferred
  from a text drawn from a process must be roughly smaller than the
  number of distinct word-like strings detected in this text by means
  of the PPM compression algorithm. We also observe that the number of
  the word-like strings for a sample of plays by Shakespeare follows
  an empirical stepwise power law, in a stark contrast to Markov
  processes. Hence we suppose that natural language considered as a
  process is not only non-Markov but also perigraphic.
  \\[2ex]
\end{abstract}


\end{titlepage}
\pagestyle{plain}   


\section{Introduction}
\subsection*{}


One of motivating assumptions of information theory
\cite{Shannon48,Shannon51,CoverThomas06} is that communication in
natural language can be reasonably modeled as a discrete stationary
stochastic process, namely, an infinite sequence of discrete random
variables with a well defined time-invariant probability
distribution. The same assumption is made in several practical
applications of computational linguistics, such as speech recognition
\cite{Jelinek97} or part-of-speech tagging
\cite{ManningSchutze99}. Whereas state-of-the-art stochastic models of
natural language are far from being satisfactory, we may ask a more
theoretically oriented question, namely:
\begin{quote}
  What can be some general mathematical properties of natural language
  treated as a stochastic process, in view of empirical data?
\end{quote}
In this paper, we will investigate a question whether it is reasonable
to assume that natural language communication is a \emph{perigraphic}
process.

To recall, a stationary process is called ergodic if the relative
frequencies of all finite substrings in the infinite text generated by
the process converge in the long run with probability one to some
constants---the probabilities of the respective strings. Now, some
basic linguistic intuition suggests that natural language does not
satisfy this property, cf. \cite[Section 6.4]{CoverThomas06}. Namely,
we can probably agree that there is a variation of topics of texts in
natural language, and these topics can be empirically distinguished by
counting relative frequencies of certain substrings called keywords.
Hence we expect that the relative frequencies of keywords in a
randomly selected text in natural language are random variables
depending on the random text topic. In the limit, for an infinitely
long text, we may further suppose that the limits of relative
frequencies of keywords persist to be random, and if this is true then
natural language is not ergodic, i.e., it is nonergodic.

In this paper we will entertain first a stronger hypothesis, namely,
that natural language communication is strongly nonergodic.
Informally speaking, a stationary process will be called strongly
nonergodic if its random persistent topic has to be described using an
infinite sequence of probabilistically independent binary random 
variables, called probabilistic facts. Like nonergodicity, strong
nonergodicity is not empirically verifiable if we only have a single
infinite sequence of data. But replacing probabilistic facts with an
algorithmically random sequence of bits, called algorithmic facts, we
can adapt the property of strong nonergodicity back to ergodic
processes.  Subsequently, we will call a process \emph{perigraphic} if
the number of algorithmic facts which can be inferred from a finite
text sampled from the process grows like a power of the text
length. It is a general observation that perigraphic processes have
uncomputable distributions.

It is interesting to note that \emph{perigraphic} processes can be
singled out by some statistical properties of the texts they generate.
We will exhibit a proposition, which we call the theorem about facts
and words. Suppose that we have a finite text drawn from a stationary
process.  The theorem about facts and words says that the number of
independent probabilistic or algorithmic facts that can be reasonably
inferred from the text must be roughly smaller than the number of
distinct word-like strings detected in the text by some standard data
compression algorithm called the Prediction by Partial Matching (PPM)
code \cite{ClearyWitten84}. It is important to stress that in this
theorem we do not relate the numbers all facts and all word-like
strings, which would sound trivial, but we compare only the numbers of
independent facts and distinct word-like strings.

Having the theorem about facts and words, we can also discuss some
empirical data.  Since the number of distinct word-like strings for
texts in natural language follows an empirical stepwise power law, in
a stark contrast to Markov processes, consequently, we suppose that
the number of inferrable random facts for natural language also
follows a power law. That is, we suppose that natural language is not
only non-Markov but also \emph{perigraphic}.

Whereas in this paper we fill several important missing gaps and
provide an overarching narration, the basic ideas presented in this
paper are not so new.  The starting point was a corollary of Zipf's
law and a hypothesis by Hilberg. Zipf's law is an empirical
observation that in texts in natural language, the frequencies of
words obey a power law decay when we sort the words according to their
decreasing frequencies \cite{Zipf65,Mandelbrot54}. A corollary of this
law, called Heaps' law
\cite{KuraszkiewiczLukaszewicz51en,Guiraud54,Herdan64,Heaps78}, states
that the number of distinct words in a text in natural language grows
like a power of the text length. In contrast to these simple empirical
observations, Hilberg's hypothesis is a less known conjecture about
natural language that the entropy of a text chunk of an increasing
length \cite{Hilberg90} or the mutual information between two adjacent
text chunks
\cite{EbelingNicolis91,EbelingPoschel94,BialekNemenmanTishby01b,CrutchfieldFeldman03}
obey also a power law growth.  In paper \cite{Debowski06}, it was
heuristically shown that if Hilberg's hypothesis for mutual
information is satisfied for an arbitrary stationary stochastic
process then texts drawn from this process satisfy also a kind of
Heaps' law if we detect the words using the grammar-based codes
\cite{Wolff80,DeMarcken96,KitWilks99,KiefferYang00}.  This result is a
historical antecedent of the theorem about facts and words.

Another important step was a discovery of some simple strongly
nonergodic processes, satisfying the power law growth of mutual
information, called Santa Fe processes, discovered by Dębowski in
August 2002, but first reported only in \cite{Debowski09}.
Subsequently, in paper \cite{Debowski11b}, a completely formal proof
of the theorem about facts and words for strictly minimal
grammar-based codes \cite{KiefferYang00,CharikarOthers05} was
provided. The respective related theory of natural language was later
reviewed in \cite{Debowski11d,Debowski15} and supplemented by a
discussion of Santa Fe processes in \cite{Debowski12}. Some drawback
of this theory at that time was that strictly minimal grammar-based
codes used in the statement of the theorem about facts and words are
not computable in a polynomial time \cite{CharikarOthers05}. This
precluded an empirical verification of the theory.

To state the relative novelty, in this paper we are glad to announce a
new stronger version of the theorem about facts and words for a
somewhat more elegant definition of inferrable facts and the PPM code,
which is computable almost in a linear time.  For the first time, we
also present two cases of the theorem: one for strongly nonergodic
processes, applying Shannon information theory, and one for general
stationary processes, applying algorithmic information theory.  Having
these results, we can supplement them finally with a rudimentary
discussion of some empirical data.

The organization of this paper is as follows.  In Section
\ref{secErgodic}, we discuss some properties of ergodic and nonergodic
processes.  In Section \ref{secStronglyNonergodic}, we define strongly
nonergodic processes and we present some examples of
them. Analogically, in Section \ref{secPerigraphic}, we discuss
perigraphic processes.  In Section \ref{secFactsWords}, we discuss two
versions of the theorem about facts and words.  In Section
\ref{secNaturalLanguage}, we discuss some empirical data and we
suppose that natural language may be a perigraphic process.  In
Section \ref{secConclusion}, we offer concluding remarks. Moreover,
three appendices follow the body of the paper. In Appendix
\ref{secFactsMI}, we prove the first part of the theorem about facts
and words.  In Appendix \ref{secMIWords}, we prove the second part of
this theorem.  In Appendix \ref{secSantaFe}, we show that that the
number of inferrable facts for the Santa Fe processes follows a power
law.

\section{Ergodic and nonergodic processes}
\label{secErgodic}

We assume that the reader is familiar with some probability measure
theory \cite{Billingsley79}.  For a real-valued random variable $Y$ on
a probability space $(\Omega,\mathcal{J},P)$, we denote its
expectation
\begin{align}
  \sred Y:=\int YdP.
\end{align}
Consider now a discrete stochastic process
$(X_i)_{i=1}^\infty=(X_1,X_2,...)$, where random variables $X_i$ take
values from a set $\mathbb{X}$ of countably many distinct symbols,
such as letters with which we write down texts in natural language. We
denote blocks of consecutive random variables $X_j^k:=(X_j,...,X_k)$
and symbols $x_j^k:=(x_j,...,x_k)$.  Let us define a binary random
variable telling whether some string $x_1^n$ has occurred in sequence
$(X_i)_{i=1}^\infty$ on positions from $i$ to $i+n-1$,
\begin{align}
  \Phi_i(x_1^n):=\boole{X_{i}^{i+n-1}=x_1^n},
\end{align}
where
\begin{align}
  \boole{\phi}
  =
  \begin{cases}
    1 & \text{if $\phi$ is true},\\
    0 & \text{if $\phi$ is false}.
  \end{cases} 
\end{align}
The expectation of this random variable, 
\begin{align}
  \sred \Phi_i(x_1^n)=P(X_{i}^{i+n-1}=x_1^n),
\end{align}
is the probability of the chosen string, whereas the arithmetic
average of consecutive random variables
$\frac{1}{m}\sum_{i=1}^m \Phi_i(x_1^n)$ is the relative frequency of the
same string in a finite sequence of random symbols $X_1^{m+n-1}$.

Process $(X_i)_{i=1}^\infty$ is called \emph{stationary} (with respect
to a probability measure $P$) if expectations $\sred \Phi_i(x_1^n)$ do
not depend on position $i$ for any string $x_1^n$. In this case, we
have the following well known theorem, which establishes that the
limiting relative frequencies of strings $x_1^n$ in infinite sequence
$(X_i)_{i=1}^\infty$ exist almost surely, i.e., with probability $1$:
\begin{theorem}[ergodic theorem, cf. e.g. \cite{Gray09}]
  For any discrete stationary process $(X_i)_{i=1}^\infty$, there
  exist limits
  \begin{align}
    \Phi(x_1^n):=\lim_{m\rightarrow\infty}
    \frac{1}{m}\sum_{i=1}^m \Phi_i(x_1^n) \text{ almost surely},
  \end{align}
  with expectations $\sred \Phi(x_1^n)=\sred \Phi_i(x_1^n)$.
\end{theorem}
In general, limits $\Phi(x_1^n)$ are random variables depending
on a particular value of infinite sequence $(X_i)_{i=1}^\infty$. It is
quite natural, however, to require that the relative frequencies of
strings $\Phi(x_1^n)$ are almost surely constants, equal to the
expectations $\sred \Phi_i(x_1^n)$.  Subsequently, process
$(X_i)_{i=1}^\infty$ will be called \emph{ergodic} (with respect to a
probability measure $P$) if limits $\Phi(x_1^n)$ are almost
surely constant for any string $x_1^n$.  The standard
definition of an ergodic process is more abstract but is equivalent to
this statement \cite[Lemma 7.15]{Gray09}.

The following examples of ergodic processes are well known:
\begin{enumerate}
\item Process $(X_i)_{i=1}^\infty$ is called \emph{IID} (independent
  identically distributed) if
  \begin{align}
    P(X_1^n=x_1^n)=\pi(x_1)...\pi(x_n).
  \end{align}
  All IID processes are ergodic.
\item Process $(X_i)_{i=1}^\infty$ is called \emph{Markov} (of order
  $1$) if
  \begin{align}
    P(X_1^n=x_1^n)=\pi(x_1)p(x_2|x_1)...p(x_n|x_{n-1}). 
  \end{align}
  A Markov process is ergodic in particular if 
  \begin{align}
    \label{MarkovMixing}
    p(x_i|x_{i-1})>c>0.
  \end{align}
  For a sufficient and necessary condition see \cite[Theorem
  7.16]{Breiman92}.
\item Process $(X_i)_{i=1}^\infty$ is called \emph{hidden Markov} if
  $X_i=g(S_i)$ for a certain Markov process $(S_i)_{i=1}^\infty$ and a
  function $g$.  A hidden Markov process is ergodic in particular if
  the underlying Markov process is ergodic.
\end{enumerate}
Whereas IID and Markov processes are some basic models in probability
theory, hidden Markov processes are of practical importance in
computational linguistics \cite{Jelinek97,ManningSchutze99}.  Hidden
Markov processes as considered there usually satisfy condition
(\ref{MarkovMixing}) and therefore they are ergodic.

Let us call a probability measure $P$ stationary or ergodic,
respectively, if the process $(X_i)_{i=1}^\infty$ is stationary or
ergodic with with respect to the measure $P$. Suppose that we have a
stationary measure $P$ which generates some data
$(X_i)_{i=1}^\infty$. We can define a new random measure $F$ equal to
the relative frequencies of blocks in the data $(X_i)_{i=1}^\infty$.
It turns out that the measure $F$ is almost surely ergodic. Formally,
we have this proposition.
\begin{theorem}[\mbox{cf. \cite[Theorem 9.10]{Kallenberg97}}]
  \label{theoPreErgodicDecomp}
  Any process $(X_i)_{i=1}^\infty$ with a stationary measure $P$ is
  almost surely ergodic with respect to the random measure $F$ given
  by
  \begin{align}
    F(X_1^n=x_1^n):=\Phi(x_1^n).  
  \end{align}
\end{theorem}
Moreover, from the random measure $F$ we can obtain the stationary
measure $P$ by integration, $P(X_1^n=x_1^n)=\sred F(X_1^n=x_1^n)$. The
following result asserts that this integral representation of measure
$P$ is unique.
\begin{theorem}[\mbox{ergodic decomposition, cf. \cite[Theorem 9.12]{Kallenberg97}}]
  Any stationary probability measure $P$ can be represented as
  \begin{align}
    P(X_1^n=x_1^n)=\int F(X_1^n=x_1^n)d\nu(F),
  \end{align}
  where $\nu$ is a unique measure on stationary ergodic measures.
\end{theorem}
In other words, stationary ergodic measures are some building blocks
from which we can construct any stationary measure. For a stationary
probability measure $P$, the particular values of the random ergodic
measure $F$ are called the ergodic components of measure $P$.

Consider for instance, a Bernoulli($\theta$) process with
measure
\begin{align}
  \label{Bernoulli}
  F_\theta(X_1^n=x_1^n)=\theta^{\sum_{i=1}^n
    x_i}(1-\theta)^{n-\sum_{i=1}^n x_i},
\end{align}
where $x_i\in\klam{0,1}$ and $\theta\in[0,1]$.  This measure will be
contrasted with the measure of a mixture Bernoulli process with 
parameter $\theta$ uniformly distributed on interval $[0,1]$,
\begin{align}
  \label{MixtureBernoulli}
  P(X_1^n=x_1^n)&=\int_0^1
  F_\theta(X_1^n=x_1^n)d\theta
  \nonumber\\
  &=\frac{1}{n+1}\kwad{\binom{n}{\sum_{i=1}^n x_i}}^{-1}.
\end{align}
Measure (\ref{Bernoulli}) is a measure of an IID process and is
therefore ergodic, whereas measure (\ref{MixtureBernoulli}) is a
mixture of ergodic measures and hence it is nonergodic.

\section{Strongly nonergodic processes}
\label{secStronglyNonergodic}

According to our definition, a process is ergodic when the relative
frequencies of any strings in a random sample in the long run converge
to some constants. Consider now the following thought
experiment. Suppose that we select a random book from a library.
Counting the relative frequencies of keywords, such as
\emph{bijection} for a text in mathematics and \emph{fossil} for a
text in paleontology, we can effectively recognize the topic of the
book. Simply put, the relative frequencies of some keywords will be
higher for books concerning some topics whereas they will be lower for
books concerning other topics. Hence, in our thought experiment, we
expect that the relative frequencies of keywords are some random
variables with values depending on the particular topic of the
randomly selected book. Since keywords are some particular strings, we
may conclude that the stochastic process that models natural language
should be nonergodic.
  
The above thought experiment provides another perspective onto
nonergodic processes. According to the following theorem, a process is
nonergodic when we can effectively distinguish in the limit at least
two random topics in it. In the statement, function
$f:\mathbb{X}^*\rightarrow\klam{0,1,2}$ assumes values $0$ or $1$ when
we can identify the topic, whereas it takes value $2$ when we are not
certain which topic a given text is about.
\begin{theorem}[cf. \cite{Debowski09}]
  \label{theoKnowability}
  A stationary discrete process $(X_i)_{i=1}^\infty$ is nonergodic if
  and only if there exists a function
  $f:\mathbb{X}^*\rightarrow\klam{0,1,2}$ and a binary random variable
  $Z$ such that $0<P(Z=0)<1$ and
  \begin{align}
    \label{Knowability}
    \lim_{n\rightarrow\infty} P(f(X_{i}^{i+n-1})=Z)=1
  \end{align}
  for any position $i\in\mathbb{N}$.
\end{theorem}
A binary variable $Z$ satisfying condition (\ref{Knowability}) will be
called a \emph{probabilistic fact}.  A probabilistic fact tells which
of two topics the infinite text generated by the stationary process is
about. It is a kind of a random switch which is preset before we start
scanning the infinite text, compare a similar wording in
\cite{GrayDavisson74b}. To keep the proofs simple, here we only give a
new elementary proof of the ``$\implies$'' statement of Theorem
\ref{theoKnowability}. The proof of the ``$\impliedby$'' part applies
some measure theory and follows the idea of Theorem 9 from
\cite{Debowski09} for strongly nonergodic processes, which we will
discuss in the next paragraph.
\begin{proof} (only $\implies$) Suppose that process
  $(X_i)_{i=1}^\infty$ is nonergodic. Then there exists a string
  $x_1^k$ such that $\Phi\neq \sred \Phi$ for $\Phi:=\Phi(x_1^k)$ with
  some positive probability. Hence there exists a real number $y$ such
  that $P(\Phi=y)=0$ and
\begin{align}
  \label{Separation}
  P(\Phi>y)=1-P(\Phi<y)\in(0,1)
  .
\end{align}
Define $Z:=\boole{\Phi>y}$ and
$f(X_{i}^{i+n-1}):=Z_{in}:=\boole{\Phi_{in}>y}$, where
\begin{align}
  \Phi_{in}:=\frac{1}{n-k+1}\sum_{j=i}^{i+n-k}
  \Phi_j(x_1^k)
  .
\end{align}
Since $\lim_{n\rightarrow\infty} \Phi_{in}=\Phi$ almost surely and
$\Phi$ satisfies (\ref{Separation}), convergence
$\lim_{n\rightarrow\infty} Z_{in}=Z$ also holds almost
surely. Applying the Lebesgue dominated convergence theorem we obtain
\begin{align}
 \lim_{n\rightarrow\infty}
 P(f(X_{i}^{i+n-1})=Z)
 &= \lim_{n\rightarrow\infty}
 \sred\kwad{Z_{in}Z+(1-Z_{in})(1-Z)}
 \nonumber\\
 &=\sred\kwad{Z^2+(1-Z)^2}=1 
 .
\end{align}
\end{proof}

As for books in natural language, we may have an intuition that the
pool of available book topics is extremely large and contains many
more topics than just two. For this reason, we may need not a single
probabilistic fact $Z$ but rather a sequence of probabilistic facts
$Z_1,Z_2,...$ to specify the topic of a book completely.  Formally,
stationary processes requiring an infinite sequence of independent
uniformly distributed probabilistic facts to describe the topic of an
infinitely long text will be called strongly nonergodic.
\begin{definition}[cf. \cite{Debowski09,Debowski11b}]
  A stationary discrete process $(X_i)_{i=1}^\infty$ is called
  \emph{strongly nonergodic} if there exist a function
  $g:\mathbb{N}\times\mathbb{X}^*\rightarrow\klam{0,1,2}$ and a~binary
  IID process $(Z_k)_{k=1}^\infty$ such that $P(Z_k=0)=P(Z_k=1)=1/2$
  and
   \begin{align}
     \label{StrongKnowability}
     \lim_{n\rightarrow\infty} P(g(k;X_{i}^{i+n-1})=Z_k)=1
   \end{align}
   for any position $i\in\mathbb{N}$ and any index $k\in\mathbb{N}$.
 \end{definition}
 As we have stated above, for a strongly nonergodic process, there is
 an infinite number of independent probabilistic facts
 $(Z_k)_{k=1}^\infty$ with a uniform distribution on the set
 $\klam{0,1}$. Formally, these probabilistic facts can be assembled into a
 single real random variable $T=\sum_{k=1}^\infty 2^{-k} Z_k$, which
 is uniformly distributed on the unit interval $[0,1]$. The value of
 variable $T$ identifies the topic of a random infinite text generated
 by the stationary process.  Thus for a strongly nonergodic process,
 we have a continuum of available topics which can be incrementally
 identified from any sufficiently long text. Put formally, according
 to Theorem 9 from \cite{Debowski09} a stationary process is strongly
 nonergodic if and only if its shift-invariant $\sigma$-field contains
 a nonatomic sub-$\sigma$-field. We note in passing that in
 \cite{Debowski09} strongly nonergodic processes were called
 \emph{uncountable description processes}.

 In view of Theorem 9 from \cite{Debowski09}, the mixture Bernoulli
 process (\ref{MixtureBernoulli}) is some example of a strongly
 nonergodic process. In this case, the parameter $\theta$ plays the
 role of the random variable $T=\sum_{k=1}^\infty 2^{-k} Z_k$.
 Showing that condition (\ref{StrongKnowability}) is satisfied for
 this process in an elementary fashion is a tedious exercise. Hence
 let us present now a simpler guiding example of a strongly nonergodic
 process, which we introduced in \cite{Debowski09,Debowski11b} and
 called the Santa Fe process.  Let $(Z_k)_{k=1}^\infty$ be a binary
 IID process with $P(Z_k=0)=P(Z_k=1)=1/2$.  Let $(K_i)_{i=1}^\infty$
 be an IID process with $K_i$ assuming values in natural numbers with
 a power-law distribution
\begin{align}
  P(K_i=k)\propto \frac{1}{k^\alpha}, \quad \alpha>1.
\end{align}
The \emph{Santa Fe process} with exponent $\alpha$ is a sequence
$(X_i)_{i=1}^\infty$, where
\begin{align}
  X_i=(K_i,Z_{K_i})
\end{align}
are pairs of a random number $K_i$ and the corresponding probabilistic fact
$Z_{K_i}$.  The Santa Fe process is strongly nonergodic since condition
(\ref{StrongKnowability}) holds for example for
\begin{align}
  \label{SantaFePredictor}
  g(k;x_1^n)
  =
  \begin{cases}
    0 & \text{if for all $1\le i\le n$, $x_i=(k,z)\implies x_i=(k,0)$},
    \\
    1 & \text{if for all $1\le i\le n$, $x_i=(k,z)\implies x_i=(k,1)$},
    \\
    2 & \text{else}.
  \end{cases}
\end{align}
Simply speaking, function $g(k;\cdot)$ returns $0$ or $1$ when an
unambiguous value of the second constituent can be read off from pairs
$x_i=(k,\cdot)$ and returns $2$ when there is some
ambiguity. Condition (\ref{StrongKnowability}) is satisfied since
\begin{align}
  P(g(k;X_{i}^{i+n-1})=Z_k)&=P(\text{$K_i=k$ for some $1\le i\le
    n$})
  \nonumber\\
  &=1-(1-P(K_i=k))^n\xrightarrow[n\rightarrow\infty]{} 1.
\end{align}

Some salient property of the Santa Fe process is the power law growth
of the expected number of probabilistic facts which can be inferred
from a finite text drawn from the process.  Consider a strongly
nonergodic process $(X_i)_{i=1}^\infty$. The set of initial
independent probabilistic facts inferrable from a finite text $X_1^n$
will be defined as
 \begin{align}
   U(X_1^n):=\klam{l\in\mathbb{N}: g(k;X_1^n)=Z_k \text{ for all
       $k\le l$}}.
 \end{align}
 In other words, we have $U(X_1^n)=\klam{1,2,...,l}$ where $l$ is the
 largest number such that $g(k;X_1^n)=Z_k$ for all $k\le l$.
 To capture the power-law growth of an arbitrary function
 $s:\mathbb{N}\rightarrow\mathbb{R}$, we will denote the Hilberg
 exponent defined
\begin{align}
  \hilberg_{n\rightarrow\infty} s(n):=\limsup_{n\rightarrow\infty}
  \frac{\log^+ s(2^n)}{\log 2^n},
\end{align}
where $\log^+ x:=\log(x+1)$ for $x\ge 0$ and $\log^+ x:=0$ for $x<0$,
cf.\ \cite{Debowski15d}.  In contrast to paper \cite{Debowski15d}, for
technical reasons, we define the Hilberg exponent only for an
exponentially sparse subsequence of terms $s(2^n)$ rather than all terms
$s(n)$. Moreover, in \cite{Debowski15d}, the Hilberg exponent was
considered only for mutual information
$s(n)=\mathbb{I}(X_1^n;X_{n+1}^{2n})$, defined later in equation
(\ref{PMI}). We observe that for the exact power law growth
$s(n)=n^\beta$ with $\beta\ge 0$ we have
$\hilberg_{n\rightarrow\infty} s(n)=\beta$. More generally, the
Hilberg exponent captures an asymptotic power-law growth of the
sequence.  As shown in Appendix \ref{secSantaFe}, for the Santa Fe
process with exponent $\alpha$ we have the asymptotic power-law growth
\begin{align}
  \label{FactsExpI}
  \hilberg_{n\rightarrow\infty} \sred\card U(X_1^n) 
  =
  1/\alpha\in(0,1)
  .
\end{align}
This property distinguishes the Santa Fe process from the mixture
Bernoulli process (\ref{MixtureBernoulli}), for which the respective
Hilberg exponent is zero, as we discuss in Section
\ref{secNaturalLanguage}.

\section{Perigraphic processes}
\label{secPerigraphic}

Is it possible to demonstrate by a statistical investigation of texts
that natural language is really strongly nonergodic and satisfies a
condition similar to (\ref{FactsExpI})? In the thought experiment
described in the beginning of the previous section we have ignored the
issue of constructing an infinitely long text. In reality, every book
with a well defined topic is finite. If we want to obtain an unbounded
collection of texts, we need to assemble a corpus of different books
and it depends on our assembling criteria whether the books in the
corpus will concern some persistent random topic.  Moreover, if we
already have a \emph{single} infinite sequence of books generated by
some stationary source and we estimate probabilities as relative
frequencies of blocks of symbols in this sequence then by Theorem
\ref{theoPreErgodicDecomp} we will obtain an ergodic probability
measure almost surely.

In this situation we may ask whether the idea of the power-law growth
of the number of inferrable probabilistic facts can be translated
somehow to the case of ergodic measures.  Some straightforward method
to apply is to replace the sequence of independent uniformly
distributed probabilistic facts $(Z_k)_{k=1}^\infty$, being random
variables, with an algorithmically random sequence of particular
binary digits $(z_k)_{k=1}^\infty$. Such digits $z_k$ will be called
\emph{algorithmic facts} in contrast to variables $Z_k$ being called
\emph{probabilistic facts}.

Let us recall some basic concepts.  For a discrete random
variable $X$, let $P(X)$ denote the random variable that takes value
$P(X=x)$ when $X$ takes value $x$. We will introduce the pointwise
entropy
\begin{align}
  \label{PEntropy}
  \mathbb{H}(X):=-\log P(X)
  ,
\end{align}
where $\log$ stands for the natural logarithm.  The prefix-free
Kolmogorov complexity $K(u)$ of a string $u$ is the length of the
shortest self-delimiting program written in binary digits that prints
out string $u$ \cite[Chapter 3]{LiVitanyi08}. $K(u)$ is the founding
concept of the algorithmic information theory and is an analogue of
the pointwise entropy. To keep our notation analogical to
(\ref{PEntropy}), we will write the algorithmic entropy
\begin{align}
  \mathbb{H}_a(u):=K(u)\log 2
  .
\end{align}

If the probability measure is computable then the algorithmic entropy
is close to the pointwise entropy. On the one hand, by the
Shannon-Fano coding for a computable probability measure, the
algorithmic entropy is less than the pointwise entropy plus a constant
which depends on the probability measure and the dimensionality of the
distribution \cite[Corollary 4.3.1]{LiVitanyi08}. Formally,
\begin{align}
  \label{ShannonFano}
  \mathbb{H}_a(X_1^n)\le \mathbb{H}(X_1^n)+2\log n+C_P,
\end{align}
where $C_P\ge 0$ is a certain constant depending on the probability
measure $P$. On the other hand, since the prefix-free Kolmogorov
complexity is also the length of a prefix-free code, we have
\begin{align}
  \label{SourceCoding}
  \sred\mathbb{H}_a(X_1^n)\ge \sred\mathbb{H}(X_1^n)
  .
\end{align}
It is also true that $\mathbb{H}_a(X_1^n)\ge \mathbb{H}(X_1^n)$ for
sufficiently large $n$ almost surely \cite[Theorem 3.1]{Barron85b}.
Thus we have shown that the algorithmic entropy is in some sense close
to the pointwise entropy, for a computable probability measure.

Next, we will discuss the difference between probabilistic and
algorithmic randomness. Whereas for an IID sequence of random
variables $(Z_k)_{k=1}^\infty$ with $P(Z_k=0)=P(Z_k=1)=1/2$ we have
\begin{align}
  \mathbb{H}(Z_1^k)=k\log 2
  ,
\end{align}
similarly an infinite sequence of binary digits $(z_k)_{k=1}^\infty$
is called algorithmically random (in the Martin-L\"of sense) when
there exists a constant $C\ge 0$ such that
\begin{align}
  \mathbb{H}_a(z_1^k)\ge k\log 2-C
\end{align}
for all $k\in\mathbb{N}$ \cite[Theorem 3.6.1]{LiVitanyi08}.  The
probability that the aforementioned sequence of random variables
$(Z_k)_{k=1}^\infty$ is algorithmically random equals $1$---for example
by \cite[Theorem 3.1]{Barron85b}, so algorithmically random sequences
are typical realizations of sequence $(Z_k)_{k=1}^\infty$.

Let $(X_i)_{i=1}^\infty$ be a stationary process.  We observe that
generalizing condition (\ref{StrongKnowability}) in an algorithmic
fashion does not make much sense. Namely, condition
\begin{align}
  \label{StrongKnowabilityAlg}
  \lim_{n\rightarrow\infty} P(g(k;X_{i}^{i+n-1})=z_k)=1
\end{align}
is trivially satisfied for any stationary process for a certain
computable function
$g:\mathbb{N}\times\mathbb{X}^*\rightarrow\klam{0,1,2}$ and an
algorithmically random sequence $(z_k)_{k=1}^\infty$. It turns out so
since there exists a computable function
$\omega:\mathbb{N}\times\mathbb{N}\rightarrow\klam{0,1}$ such that
$\lim_{n\rightarrow\infty}\omega(k;n)=\Omega_k$, where
$(\Omega_k)_{k=1}^\infty$ is the binary expansion of the halting
probability $\Omega=\sum_{k=1}^\infty 2^{-k}\Omega_k$, which is a
lower semi-computable algorithmically random sequence \cite[Section
3.6.2]{LiVitanyi08}.

In spite of this negative result, the power-law growth of the number
of inferrable algorithmic facts corresponds to some nontrivial
property. For a computable function
$g:\mathbb{N}\times\mathbb{X}^*\rightarrow\klam{0,1,2}$ and an
algorithmically random sequence of binary digits $(z_k)_{k=1}^\infty$,
which we will call \emph{algorithmic facts}, the set of initial
algorithmic facts inferrable from a finite text $X_1^n$ will be
defined as
\begin{align}
  U_a(X_1^n):=\klam{l\in\mathbb{N}: g(k;X_1^n)=z_k \text{ for all
  $k\le l$}}
  .
\end{align}
Subsequently, we will call a process perigraphic if the expected
number of algorithmic facts which can be inferred from a finite text
sampled from the process grows asymptotically like a power of the text
length.
\begin{definition}
  A stationary discrete process $(X_i)_{i=1}^\infty$ is called
  \emph{perigraphic} if 
  \begin{align}
    \label{Perigraphic}
    \hilberg_{n\rightarrow\infty} \sred\card U_a(X_1^n) >0
 \end{align}
 for some computable function
 $g:\mathbb{N}\times\mathbb{X}^*\rightarrow\klam{0,1,2}$ and an
 algorithmically random sequence of binary digits
 $(z_k)_{k=1}^\infty$.
\end{definition}
Perigraphic processes can be ergodic.  The proof of Theorem
\ref{theoFacts} from Appendix \ref{secSantaFe} can be easily adapted
to show that some example of a perigraphic process is the Santa Fe
process with sequence $(Z_k)_{k=1}^\infty$ replaced by an
algorithmically random sequence of binary digits
$(z_k)_{k=1}^\infty$. This process is IID and hence ergodic.

We can also easily show the following proposition.
 \begin{theorem}
   Any perigraphic process $(X_i)_{i=1}^\infty$ has an uncomputable
   measure $P$.
 \end{theorem}
 \begin{proof}
   Assume that a perigraphic process $(X_i)_{i=1}^\infty$ has a
   computable measure $P$.  By the proof of Theorem
   \ref{theoFactsMIAlg} from Appendix \ref{secFactsMI}, we have
   \begin{align}
     \hilberg_{n\rightarrow\infty} \sred \card U_a(X_1^n)
     &\le
       \hilberg_{n\rightarrow\infty}\sred
       \kwad{\mathbb{H}_a(X_1^n)-\mathbb{H}(X_1^n)}
       .
   \end{align}
   Since for a computable measure $P$ we have inequality
   (\ref{ShannonFano}) then
   \begin{align}
     \hilberg_{n\rightarrow\infty} \sred \card U_a(X_1^n)=0.
   \end{align}
   Since we have obtained a contradiction with the assumption that the
   process is perigraphic, measure $P$ cannot be computable.
 \end{proof}
 
\section{Theorem about facts and words}
\label{secFactsWords}

In this section, we will present a result about stationary processes,
which we call the theorem about facts and words. That proposition
states that the expected number of independent probabilistic or
algorithmic facts inferrable from the text drawn from a stationary
process must be roughly less than the expected number of distinct
word-like strings detectable in the text by a simple procedure
involving the PPM compression algorithm. This result states, in
particular, that an asymptotic power law growth of the number of
inferrable probabilistic or algorithmic facts as a function of the
text length produces a statistically measurable effect, namely, an
asymptotic power law growth of the number of word-like strings.

To state the theorem about facts and words formally, we need first to
discuss the PPM code. Let us denote strings of symbols
$x_j^k:=(x_j,...,x_k)$, adopting an important convention that $x_j^k$
is the empty string for $k<j$.  In the following, we consider strings
over a finite alphabet, say, $x_i\in\mathbb{X}=\klam{1,...,D}$.  We
define the frequency of a substring $w_1^k$ in a string $x_1^n$ as
\begin{align}
  N(w_1^k|x_1^n):=\sum_{i=1}^{n-k+1}\boole{x_i^{i+k-1}=w_1^k}.
\end{align}
Now we may define the Prediction by Partial Matching (PPM)
probabilities.
\begin{definition}[cf. \cite{ClearyWitten84}]
  For $x_1^n\in\mathbb{X}^n$ and $k\in\klam{-1,0,1,...}$, we put
  \begin{align}
    \PPM_k(x_i|x_1^{i-1})&:=
    \begin{cases}
      \displaystyle\frac{1}{D}, & i\le k,
      \\
      \displaystyle\frac{N(x_{i-k}^i|x_1^{i-1})+1}{N(x_{i-k}^{i-1}|x_1^{i-2})+D},
      & i> k.
    \end{cases}
  \end{align}
  Quantity $\PPM_k(x_i|x_1^{i-1})$ is called the \emph{conditional PPM
    probability} of order $k$ of symbol $x_i$ given string
  $x_1^{i-1}$. Next, we put
  \begin{align}
    \PPM_k(x_1^n)&:=\prod_{i=1}^n\PPM_k(x_i|x_1^{i-1}).
  \end{align}
  Quantity $\PPM_k(x_1^n)$ is called the \emph{PPM probability} of
  order $k$ of string $x_1^n$. Finally, we put
  \begin{align}
    \label{PPM}
    \PPM(x_1^n)&:=\frac{6}{\pi^2}\sum_{k=-1}^\infty
    \frac{\PPM_k(x_1^n)}{(k+2)^2}.
  \end{align}
  Quantity $\PPM(x_1^n)$ is called the (total) \emph{PPM probability} of the
  string $x_1^n$. 
\end{definition}

Quantity $\PPM_k(x_1^n)$ is an incremental approximation of the
unknown true probability of the string $x_1^n$, assuming that the
string has been generated by a Markov process of order $k$.  In
contrast, quantity $\PPM(x_1^n)$ is a mixture of such Markov
approximations for all finite orders.  In general, the PPM
probabilities are probability distributions over strings of a fixed
length.  That is:
\begin{itemize}
\item $\PPM_k(x_i|x_1^{i-1})> 0$ and
$\sum_{x_i\in\mathbb{X}}\PPM_k(x_i|x_1^{i-1})=1$,
\item $\PPM_k(x_1^n)> 0$ and $\sum_{x_1^n\in\mathbb{X}^n}\PPM_k(x_1^n)=1$,
\item $\PPM(x_1^n)> 0$ and $\sum_{x_1^n\in\mathbb{X}^n}\PPM(x_1^n)=1$. 
\end{itemize}

In the following, we define an analogue of the pointwise entropy
\begin{align}
  \mathbb{H}_{\PPM}(x_1^n)&:=-\log\PPM(x_1^n).
\end{align}
Quantity $\mathbb{H}_{\PPM}(x_1^n)$ will be called the length of the
PPM code for the string $x_1^n$.  By nonnegativity of the
Kullback-Leibler divergence, we have for any random block $X_1^n$ that
\begin{align}
  \label{SourceCodingPPM}
  \sred\mathbb{H}_{\PPM}(X_1^n)&\ge \sred\mathbb{H}(X_1^n)
                                 .
\end{align}

The length of the PPM code or the PPM probability repsectively have
two notable properties. First, the PPM probability is a universal
probability, i.e., in the limit, the length of the PPM code
consistently estimates the entropy rate of a stationary source.
Second, the PPM probability can be effectively computed, i.e., the
summation in definition (\ref{PPM}) can be rewritten as a finite
sum. Let us state these two results formally.
\begin{theorem}[cf. \cite{Ryabko10}]
  \label{theoPPMUniversal}
  The PPM probability is universal in expectation, i.e., we have
  \begin{align}
    \label{PPMUniversal}
    \lim_{n\rightarrow\infty}\frac{1}{n}
    \sred\mathbb{H}_{\PPM}(X_1^n)
    &=
      \lim_{n\rightarrow\infty}\frac{1}{n}
    \sred\mathbb{H}(X_1^n)
  \end{align}
  for any stationary process $(X_i)_{i=1}^\infty$.
\end{theorem}
\begin{proof}
  For stationary ergodic processes the above claim follows by an
  iterated application of the ergodic theorem as shown in Theorem 1.1
  from \cite{Ryabko10} for so called measure $R$, which is a slight
  modification of the PPM probability. To generalize the claim for
  nonergodic processes, one can use the ergodic decomposition theorem
  but the exact proof requires a too large theoretical overload to be
  presented within the framework of this paper.
\end{proof}
\begin{theorem}
  \label{theoPPMFinite}
  The PPM probability can be effectively computed, i.e., we have
  \begin{align}
    \label{PPMFinite}
    \PPM(x_1^n)=\frac{6}{\pi^2}\sum_{k=0}^{L(x_1^n)}
    \frac{\PPM_k(x_1^n)}{(k+2)^2}+
    \okra{1-\frac{6}{\pi^2}\sum_{k=0}^{L(x_1^n)}\frac{1}{(k+2)^2}}D^{-n},
  \end{align}
  where
  \begin{align}
    L(x_1^n)=\max \klam{k: \text{$N(w_1^k|x_1^n)> 1$ for some
        $w_1^k$}}
  \end{align}
  is the maximal repetition of string $x_1^n$.
\end{theorem}
\begin{proof}
  We have $N(x_{i-k}^{i-1}|x_1^{i-2})=0$ for $k>L(x_1^i)$. Hence
  $\PPM_k(x_1^n)=D^{-n}$ for $k>L(x_1^n)$ and in view of this we
  obtain the claim.
\end{proof}
Maximal repetition as a function of a string was studied, e.g., in
\cite{DeLuca99,Debowski15f}.  Since the PPM probability is a
computable probability distribution then by (\ref{ShannonFano}) for a
certain constant $C_{\PPM}$ we have
\begin{align}
  \label{ShannonFanoPPM}
  \mathbb{H}_a(X_1^n)\le \mathbb{H}_{\PPM}(X_1^n)+2\log n+C_{\PPM}. 
\end{align}

Let us denote the length of the PPM code of order $k$,
  \begin{align}
    \mathbb{H}_{\PPM_k}(x_1^n)&:=-\log\PPM_k(x_1^n).
  \end{align}
  As we can easily see, the code length $\mathbb{H}_{\PPM}(x_1^n)$ is
  approximately equal to the minimal code length
  $\mathbb{H}_{\PPM_k}(x_1^n)$ where the minimization goes over
  $k\in\klam{-1,0,1,...}$. Thus it is meaningful to consider this
  definition of the PPM order of an arbitrary string.
\begin{definition}
  The \emph{PPM order} $G_{\PPM}(x_1^n)$ is the smallest $G$ such that
  \begin{align}
    \mathbb{H}_{\PPM_G}(x_1^n)\le \mathbb{H}_{\PPM_k}(x_1^n) \text{ for all $k\ge -1$}.
  \end{align}
\end{definition}
\begin{theorem}
  We have $G_{\PPM}(x_1^n)\le L(x_1^n)$.
\end{theorem}
\begin{proof}
 Follows by $\PPM_k(x_1^n)=D^{-n}=\PPM_{-1}(x_1^n)$ for $k>L(x_1^n)$.
\end{proof}

Let us divert for a short while from the PPM code definition.  The set
of distinct substrings of length $m$ in string $x_1^n$ is
\begin{align}
  V(m|x_1^n):=\klam{y_1^m:x_{t+1}^{t+m}=y_1^m
    \text{ for some $0\le t\le n-m$}}.
\end{align}
The cardinality of set $V(m|x_1^n)$ as a function of substring length
$m$ is called the subword complexity of string $x_1^n$
\cite{DeLuca99}. Now let us apply the concept of the PPM order to
define some special set of substrings of an arbitrary string $x_1^n$.
The set of distinct PPM words detected in $x_1^n$ will be defined as
the set $V(m|x_1^n)$ for $m=G_{\PPM}(x_1^n)$, i.e.,
\begin{align}
  \label{PPMVocab}
  V_{\PPM}(x_1^n):=V(G_{\PPM}(X_1^n)|x_1^n).
\end{align} 

Let us define the pointwise mutual information
\begin{align}
  \label{PMI}
  \mathbb{I}(X;Y):=\mathbb{H}(X)+\mathbb{H}(Y)-\mathbb{H}(X,Y)
\end{align}
and the algorithmic mutual information
\begin{align}
  \mathbb{I}_a(u;v):=\mathbb{H}_a(u)+\mathbb{H}_a(v)-\mathbb{H}_a(u,v).
\end{align}
Now we may write down the theorem about facts and words. The theorem
states that the Hilberg exponent for the expected number of initial
independent inferrable facts is less than the Hilberg exponent for the
expected mutual information and this is less than the Hilberg exponent
for the expected number of distinct detected PPM words plus the PPM
order. (The PPM order is usually much less than the number of distinct
PPM words.)
\begin{theorem}[facts and words I, cf.\ \cite{Debowski11b}]
  \label{theoFactsWords}
  Let $(X_i)_{i=1}^\infty$ be a stationary strongly nonergodic process
  over a finite alphabet.  We have inequalities
  \begin{align}
    \hilberg_{n\rightarrow\infty} \sred \card U(X_1^n) &\le
    \hilberg_{n\rightarrow\infty}\sred \mathbb{I}(X_1^n;X_{n+1}^{2n})
    \nonumber\\
    &\le \hilberg_{n\rightarrow\infty} \sred
    \kwad{G_{\PPM}(X_1^n)+\card V_{\PPM}(X_1^n)} .
    \label{FactsWords}
 \end{align}
\end{theorem}
\begin{proof}
  The claim follows by conjunction of Theorem \ref{theoFactsMI} from
  Appendix \ref{secFactsMI} and Theorem \ref{theoMIWords} from
  Appendix \ref{secMIWords}.
\end{proof}
Theorem \ref{theoFactsWords} has also an algorithmic version, for
ergodic processes in particular.
\begin{theorem}[facts and words II]
  \label{theoFactsWordsAlg}
  Let $(X_i)_{i=1}^\infty$ be a stationary process over a finite
  alphabet.  We have inequalities
  \begin{align}
    \hilberg_{n\rightarrow\infty} \sred \card U_a(X_1^n) &\le
    \hilberg_{n\rightarrow\infty}\sred \mathbb{I}_a(X_1^n;X_{n+1}^{2n})
    \nonumber\\
    &\le \hilberg_{n\rightarrow\infty} \sred
    \kwad{G_{\PPM}(X_1^n)+\card V_{\PPM}(X_1^n)} .
    \label{FactsWordsAlg}
 \end{align}
\end{theorem}
\begin{proof}
  The claim follows by conjunction of Theorem \ref{theoFactsMIAlg} from
  Appendix \ref{secFactsMI} and Theorem \ref{theoMIWords} from
  Appendix \ref{secMIWords}.
\end{proof}

The theorem about facts and words previously proven in
\cite{Debowski11b} differs from Theorem \ref{theoFactsWords} in three
aspects. First of all, the theorem in \cite{Debowski11b} did not apply
the concept of the Hilberg exponent and compared
$\liminf_{n\rightarrow\infty}$ with $\limsup_{n\rightarrow\infty}$
rather than $\limsup_{n\rightarrow\infty}$ with
$\limsup_{n\rightarrow\infty}$. Second, the number of inferrable facts
was defined as a functional of the process distribution rather than a
random variable depending on a particular text. Third, the number of
words was defined using a minimal grammar-based code rather than the
concept of the PPM order.  Minimal grammar-based codes are not
computable in a polynomial time in contrast to the PPM order.  Thus we
may claim that Theorem \ref{theoFactsWords} is stronger than the
theorem about facts and words previously proven in
\cite{Debowski11b}. Moreover, applying Kolmogorov complexity and
algorithmic randomness to formulate and prove Theorem
\ref{theoFactsWordsAlg} is a new idea.

It is an interesting question whether we have an almost sure version
of Theorems \ref{theoFactsWords} and \ref{theoFactsWordsAlg}, namely,
whether
\begin{align}
  \hilberg_{n\rightarrow\infty} \card U(X_1^n)
  &\le
    \hilberg_{n\rightarrow\infty} \mathbb{I}(X_1^n;X_{n+1}^{2n})
    \nonumber\\
  &\le \hilberg_{n\rightarrow\infty} \kwad{G_{\PPM}(X_1^n)+\card
    V_{\PPM}(X_1^n)} \text{ almost surely}
    \label{FactsWordsAS}
\end{align}
for strongly nonergodic processes, or
\begin{align}
  \hilberg_{n\rightarrow\infty} \card U_a(X_1^n)
  &\le
    \hilberg_{n\rightarrow\infty} \mathbb{I}_a(X_1^n;X_{n+1}^{2n})
    \nonumber\\
  &\le \hilberg_{n\rightarrow\infty} \kwad{G_{\PPM}(X_1^n)+\card
    V_{\PPM}(X_1^n)} \text{ almost surely}
    \label{FactsWordsAS}
\end{align}
for general stationary processes.  We leave this question as an open
problem.

\section{Hilberg exponents and empirical data}
\label{secNaturalLanguage}

It is advisable to show that the Hilberg exponents considered in
Theorem \ref{theoFactsWords} can assume any value in range $[0,1]$ and
the difference between them can be arbitrarily large. We adopt a
convention that the set of inferrable probabilistic facts is empty for
ergodic processes, $U(X_1^n)=\emptyset$. With this remark in mind, let
us inspect some examples of processes.

First of all, for Markov processes and their strongly nonergodic
mixtures, of any order $k$ but over a finite alphabet, we have
\begin{align}
  \hilberg_{n\rightarrow\infty} \sred\card U(X_1^n) 
  =
  \hilberg_{n\rightarrow\infty} \sred \mathbb{I}(X_1^n;X_{n+1}^{2n}) 
  =
  0
  .
\end{align}
This happens to be so since the sufficient statistic of text $X_1^n$
for predicting text $X_{n+1}^{2n}$ is the maximum likelihood estimate
of the transition matrix, the elements of which can assume at most
$(n+1)$ distinct values. Hence
$\sred \mathbb{I}(X_1^n;X_{n+1}^{2n})\le D^{k+1}\log (n+1)$, where $D$ is
the cardinality of the alphabet and $k$ is the Markov order of the
process.  Similarly, it can be shown for these processes that the PPM
order satisfies $\lim_{n\rightarrow\infty}G_{\PPM}(X_1^n)\le k$. Hence
the number of PPM words, which satisfies inequality
$\card V_{\PPM}(X_1^n)\le D^{G_{\PPM}(X_1^n)}$, is also bounded
above. In consequence, for Markov processes and their strongly
nonergodic mixtures, of any order but over a finite alphabet, we
obtain
\begin{align}
  \label{FiniteMarkovHilberg}
  \hilberg_{n\rightarrow\infty}  \kwad{G_{\PPM}(X_1^n)+\card V_{\PPM}(X_1^n)}
  =0 \text{ almost surely}.
\end{align}

In contrast, Santa Fe processes are strongly nonergodic mixtures of
some IID processes over an infinite alphabet. Being mixtures of IID
processes over an infinite alphabet, they need not satisfy condition
(\ref{FiniteMarkovHilberg}). In fact, as shown in
\cite{Debowski11b,Debowski12} and Appendix \ref{secSantaFe}, for the
Santa Fe process with exponent $\alpha$ we have the asymptotic
power-law growth
\begin{align}
  \hilberg_{n\rightarrow\infty} \sred\card U(X_1^n) 
  =
  \hilberg_{n\rightarrow\infty} \sred \mathbb{I}(X_1^n;X_{n+1}^{2n}) 
  =
  1/\alpha\in(0,1)
  .
\end{align}
The same equality for the number of inferrable probabilistic facts and
the mutual information is also satisfied by a stationary coding of the
Santa Fe process into a finite alphabet, see \cite{Debowski12}.

Let us also note that, whereas the theorem about facts and words
provides an inequality of Hilberg exponents, this inequality can be
strict. To provide some substance, in \cite{Debowski12}, we have
constructed a modification of the Santa Fe process which is ergodic
and over a finite alphabet. For this modification, we have only the
power-law growth of mutual information
\begin{align}
  \hilberg_{n\rightarrow\infty}
  \sred \mathbb{I}(X_1^n;X_{n+1}^{2n}) &= 1/\alpha\in(0,1).
\end{align}
Since in this case,
$\hilberg_{n\rightarrow\infty} \sred\card U(X_1^n)=0$ then the
difference between the Hilberg exponents for the number of inferrable
probabilistic facts and the number of PPM words can be an arbitrary
number in range $(0,1)$.

Now we are in a position to discuss some empirical data.  In this
case, we cannot directly measure the number of facts and the mutual
information but we can compute the PPM order and count the number of
PPM words.  In Figure~\ref{figPPMVocabulary}, we have presented data
for a collection of 35 plays by William
Shakespeare\footnote{Downloaded from the Project Gutenberg,
  \url{https://www.gutenberg.org/}.} and a random permutation of
characters appearing in this collection of texts.  The random
permutation of characters is an IID process over a finite alphabet so
in this case we obtain
\begin{align}
  \label{WordsHilbergZero}
  \hilberg_{n\rightarrow\infty} \card V_{\PPM}(x_1^n)=0
  .
\end{align}
In contrast, for the plays of Shakespeare we seem to have a stepwise
power law growth of the number of distinct PPM words. Thus we may
suppose that for natural language we have more generally
\begin{align}
  \label{WordsHilbergNL}
  \hilberg_{n\rightarrow\infty} \card V_{\PPM}(x_1^n)>0
  .
\end{align}
If relationship (\ref{WordsHilbergNL}) holds true then natural
language cannot be a Markov process of any order.  Moreover, in view
of the striking difference between observations
(\ref{WordsHilbergZero}) and (\ref{WordsHilbergNL}), we may suppose
that the number of inferrable probabilistic or algorithmic facts for
texts in natural language also obeys a power-law growth. Formally
speaking, this condition would translate to natural language being
strongly nonergodic or perigraphic.  We note that this hypothesis
arises only as a form of a weak inductive inference since formally we
cannot deduce condition (\ref{Perigraphic}) from mere condition
(\ref{WordsHilbergNL}), regardless of the amount of data supporting
condition (\ref{WordsHilbergNL}).

\begin{figure}[p]
  \centering
  \includegraphics[width=0.9\textwidth]{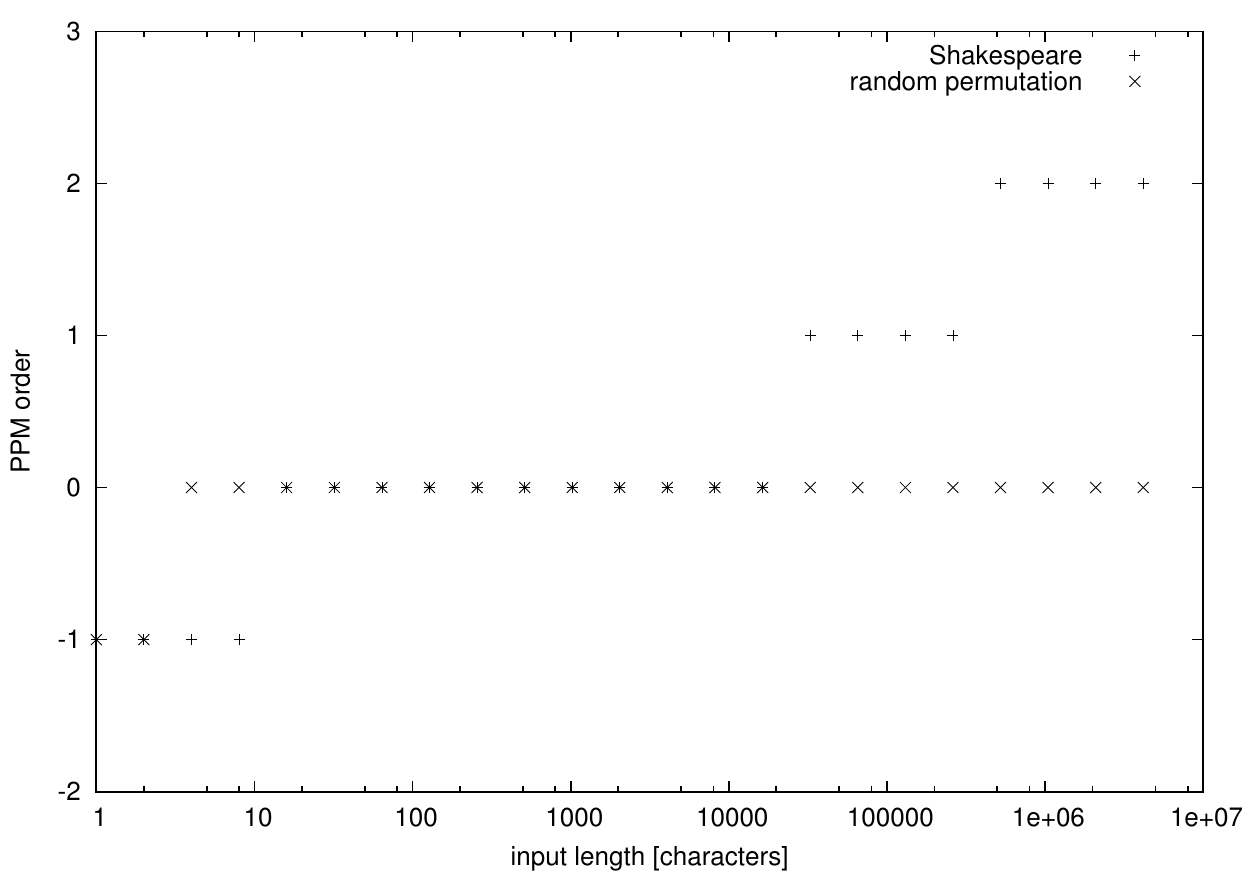}
  \centering
  \includegraphics[width=0.9\textwidth]{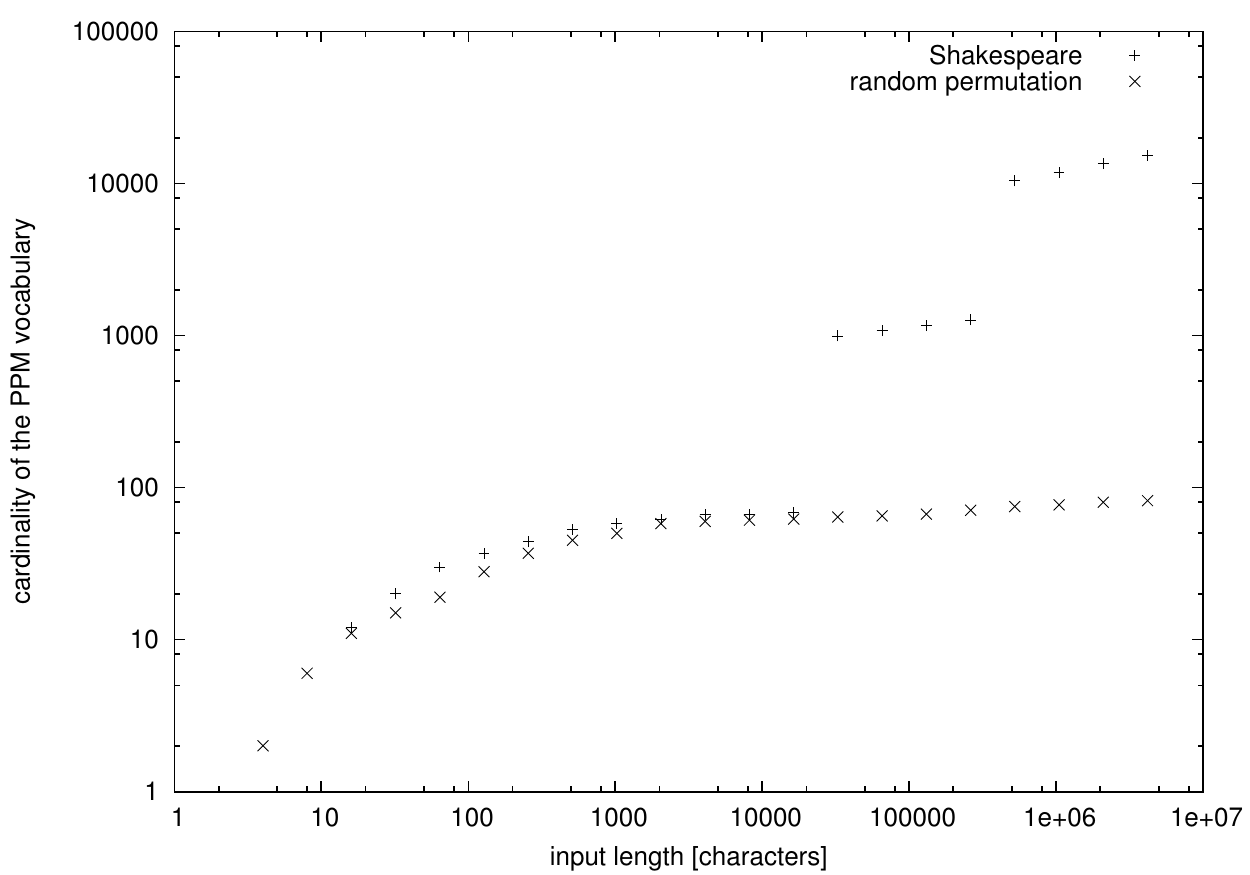}
  \caption{\label{figPPMVocabulary} The PPM order $G_{\PPM}(x_1^n)$
    and the cardinality of the PPM vocabulary $\card V_{\PPM}(x_1^n)$
    versus the input length $n$ for William Shakespeare's First
    Folio/35 Plays and a random permutation of the
    text's characters.}
\end{figure}

\section{Conclusion}
\label{secConclusion}

In this article, a stationary process has been called strongly
nonergodic if some persistent random topic can be detected in the
process and an infinite number of independent binary random variables,
called probabilistic facts, is needed to describe this topic
completely. Replacing probabilistic facts with an algorithmically
random sequence of bits, called algorithmic facts, we have adapted
this property back to ergodic processes. Subsequently, we have called
a process perigraphic if the number of algorithmic facts which can be
inferred from a finite text sampled from the process grows like a
power of the text length.

We have demonstrated an assertion, which we call the theorem about
facts and words. This proposition states that the number of
independent probabilistic or algorithmic facts which can be inferred
from a text drawn from a process must be roughly smaller than the
number of distinct word-like strings detected in this text by means of
the PPM compression algorithm.  We have exhibited two versions of this
theorem: one for strongly nonergodic processes, applying the Shannon
information theory, and one for ergodic processes, applying the
algorithmic information theory.

Subsequently, we have exhibited an empirical observation that the
number of distinct word-like strings grows like a stepwise power law
for a collections of plays by William Shakespeare, in a stark contrast
to Markov processes. This observation does not rule out that the
number of probabilistic or algorithmic facts inferrable from texts in
natural language also grows like a power law. Hence we have supposed
that natural language is a perigraphic process.
  
We suppose that the path of the future related research should lead
through a further analysis of the theorem about facts and words and
demonstrating an almost sure version of this statement.

\section*{Acknowlegdment}

We wish to thank Jacek Koronacki, Jan Mielniczuk, and Vladimir Vovk
for helpful comments.

\appendix

\section{Facts and mutual information}
\label{secFactsMI}

In the appendices, we will make use of several kinds of information measures.
\begin{enumerate}
\item First, there are four pointwise Shannon information measures:
\begin{itemize}
\item entropy \\ $\mathbb{H}(X)=-\log P(X)$,
\item conditional entropy \\ $\mathbb{H}(X|Z):=-\log P(X|Z)$,
\item mutual information \\
  $\mathbb{I}(X;Y):=\mathbb{H}(X)+\mathbb{H}(Y)-\mathbb{H}(X,Y)$,
\item conditional mutual information \\
  $\mathbb{I}(X;Y|Z):=\mathbb{H}(X|Z)+\mathbb{H}(Y|Z)-\mathbb{H}(X,Y|Z)$,
\end{itemize}
where $P(X)$ is the probability of a random variable $X$ and $P(X|Z)$
is the conditional probability of a random variable $X$ given a random
variable $Z$.  The above definitions make sense for discrete-valued
random variables $X$ and $Y$ and an arbitrary random variable $Z$.  If
$Z$ is a discrete-valued random variable then also
$\mathbb{H}(X,Z)-\mathbb{H}(Z)=\mathbb{H}(X|Z)$ and
$\mathbb{I}(X;Z)=\mathbb{H}(X)-\mathbb{H}(X|Z)$.
\item 
Moreover, we will use four algorithmic information measures:
\begin{itemize}
\item entropy \\ $\mathbb{H}_a(x)=K(x)\log 2$,
\item conditional entropy \\
  $\mathbb{H}_a(x|z):=K(x|z)\log 2$,
\item mutual information \\
  $\mathbb{I}_a(x;y):=\mathbb{H}_a(x)+\mathbb{H}_a(y)-\mathbb{H}_a(x,y)$,
\item conditional mutual information \\
  $\mathbb{I}_a(x;y|z):=\mathbb{H}_a(x|z)+\mathbb{H}_a(y|z)-\mathbb{H}_a(x,y|z)$,
\end{itemize}
where $K(x)$ is the prefix-free Kolmogorov complexity of an object $x$
and $K(x|z)$ is the prefix-free Kolmogorov complexity of an object $x$
given an object $z$.  In the above definitions, $x$ and $y$ must be
finite objects (finite texts), whereas $z$ can be also an infinite
object (an infinite sequence). If $z$ is a finite object then
$\mathbb{H}_a(x,z)-\mathbb{H}_a(z)\peq \mathbb{H}_a(x|z,K(z))$ rather
than being equal to $\mathbb{H}_a(x|z)$, where $\peq$, $\ple$, and
$\pge$ are the equality and the inequalities up to an additive
constant \cite[Theorem 3.9.1]{LiVitanyi08}. Hence
\begin{align}
  \mathbb{H}_a(x)-\mathbb{H}_a(x|z)+\mathbb{H}_a(K(z))
  &\pge
    \mathbb{I}_a(x;z)\peq \mathbb{H}_a(x)-\mathbb{H}_a(x|z,K(z))
    \nonumber\\
    &\pge
  \mathbb{H}_a(x)-\mathbb{H}_a(x|z)
  .
\end{align}
\end{enumerate}

In the following, we will prove a result for Hilberg exponents.
\begin{theorem}
  \label{theoHilbergRedundancy}
  Define $\mathfrak{J}(n):=2\mathfrak{G}(n)-\mathfrak{G}(2n)$. If the limit
  $\lim_{n\rightarrow\infty} \mathfrak{G}(n)/n=\mathfrak{g}$ exists
  and is finite then
\begin{align}
  \label{HilbergRedundancyI}
  \hilberg_{n\rightarrow\infty} \kwad{\mathfrak{G}(n)-n\mathfrak{g}}
  \le
  \hilberg_{n\rightarrow\infty} \mathfrak{J}(n)
  , 
\end{align}
with an equality if $\mathfrak{J}(2^n)\pge 0$ for all but finitely many $n$.
\end{theorem}
\begin{proof}
  The proof makes use of the telescope sum 
  \begin{align}
    \label{ESeriesRedundancy}
    \sum_{k=0}^\infty \frac{\mathfrak{J}(2^{k+n})}{2^{k+1}}
    =
    \mathfrak{G}(2^n)-2^n\mathfrak{g}
    .
  \end{align}
  Denote $\delta:=\hilberg_{n\rightarrow\infty}
  \mathfrak{J}(n)$. Since $\hilberg_{n\rightarrow\infty}
  \okra{\mathfrak{G}(n)-n\mathfrak{g}}\le 1$, it is sufficient to
  prove inequality (\ref{HilbergRedundancyI}) for $\delta<1$.  In this
  case, $\mathfrak{J}(2^n)\le 2^{(\delta+\epsilon)n}$ for all but finitely
  many $n$ for any $\epsilon>0$. Then for $\epsilon< 1-\delta$, by the
  telescope sum (\ref{ESeriesRedundancy}) we obtain for sufficiently
  large $n$ that
  \begin{align}
    \mathfrak{G}(2^n)-2^n\mathfrak{g}\le \sum_{k=0}^\infty
    \frac{2^{(\delta+\epsilon)(k+n)}}{2^{k+1}}
    \le
    2^{(\delta+\epsilon)n}\sum_{k=0}^\infty2^{(\delta+\epsilon-1)k-1}
    =\frac{2^{(\delta+\epsilon)n}}{2(1-2^{\delta+\epsilon-1})}
    .
  \end{align}
  Since $\epsilon$ can be taken arbitrarily small, we obtain
  (\ref{HilbergRedundancyI}).

  Now assume that $\mathfrak{J}(2^n)\pge 0$ for all but finitely many $n$.
  By the telescope sum (\ref{ESeriesRedundancy}), we have
  $\mathfrak{J}(2^n)/2\ple\mathfrak{G}(2^n)-2^n\mathfrak{g}$ for sufficiently
  large $n$.  Hence
  \begin{align}
    \delta\le \hilberg_{n\rightarrow\infty}
    \okra{\mathfrak{G}(n)-n\mathfrak{g}}
  \end{align}
  Combining this with (\ref{HilbergRedundancyI}), we obtain
  $\hilberg_{n\rightarrow\infty}
  \okra{\mathfrak{G}(n)-n\mathfrak{g}}=\delta$.
\end{proof}

For any stationary process $(X_i)_{i=1}^\infty$ over a finite alphabet
there exists a limit
\begin{align}
  \label{EntropyRate}
  h&:=\lim_{n\rightarrow\infty}\frac{\sred\mathbb{H}(X_1^n)}{n}=\sred\mathbb{H}(X_1|X_2^\infty)
     ,
\end{align}
called the entropy rate of process $(X_i)_{i=1}^\infty$
\cite{CoverThomas06}. By (\ref{SourceCoding}),
(\ref{PPMUniversal}), and (\ref{ShannonFanoPPM}), we also have
\begin{align}
  h&=\lim_{n\rightarrow\infty}\frac{\sred\mathbb{H}_a(X_1^n)}{n}
     .
\end{align}
Moreover, for a stationary process, the mutual
information satisfies
\begin{align}
  \sred\mathbb{I}(X_1^n;X_{n+1}^{2n})&=2\sred\mathbb{H}(X_1^n)-\sred\mathbb{H}(X_1^{2n})\ge 0,
  \\
  \sred\mathbb{I}_a(X_1^n;X_{n+1}^{2n})&=2\sred\mathbb{H}_a(X_1^n)-\sred\mathbb{H}_a(X_1^{2n})\pge 0.  
\end{align}
Hence by Theorem \ref{theoHilbergRedundancy}, we obtain
\begin{align}
  \label{RedundancyMI}
  \hilberg_{n\rightarrow\infty} \kwad{\sred\mathbb{H}(X_1^n)-hn}
  &= \hilberg_{n\rightarrow\infty} \sred\mathbb{I}(X_1^n;X_{n+1}^{2n}),
  \\
  \label{RedundancyMIAlg}
  \hilberg_{n\rightarrow\infty} \kwad{\sred\mathbb{H}_a(X_1^n)-hn}
  &= \hilberg_{n\rightarrow\infty} \sred\mathbb{I}_a(X_1^n;X_{n+1}^{2n}).
\end{align}

Subsequently, we will prove the initial parts of Theorems
\ref{theoFactsWords} and \ref{theoFactsWordsAlg}, i.e., the two
versions of the theorem about facts and words. The probabilistic
statement for strongly nonergodic processes goes first.
\begin{theorem}[facts and mutual information I]
  \label{theoFactsMI}
  Let $(X_i)_{i=1}^\infty$ be a stationary strongly nonergodic process
  over a finite alphabet.  We have inequality
  \begin{align}
    \hilberg_{n\rightarrow\infty} \sred \card U(X_1^n) &\le
    \hilberg_{n\rightarrow\infty} \sred\mathbb{I}(X_1^n;X_{n+1}^{2n})
                                                         .
    \label{FactsMI}
 \end{align}
\end{theorem}
\begin{proof}
  Let us write $S_n:=\card U(X_1^n)$. Observe that
  \begin{align}
  \sred\mathbb{H}(Z_1^{S_n}|S_n)&=-\sum_{s,w} P(S_n=s,Z_1^s=w)\log P(Z_1^s=w|S_n=s)
  \nonumber\\  
  &\ge -\sum_{s,w} P(S_n=s,Z_1^s=w)\log \frac{P(Z_1^s=w)}{P(S_n=s)}
  \nonumber\\  
  &= -\sum_{s,w} P(S_n=s,Z_1^s=w)\log \frac{2^{-s}}{P(S_n=s)}
  \nonumber\\  
                  &= (\log 2)\sred S_n-\sred\mathbb{H}(S_n),
                    \label{HZS}
  \\
  \sred\mathbb{H}(S_n)&\le (\sred S_n+1)\log(\sred S_n+1)-\sred S_n\log\sred S_n
  \nonumber\\
  &=\log(\sred S_n+1)+\sred S_n\log\frac{\sred S_n+1}{\sred S_n}
  \nonumber\\
                  &\le \log(\sred S_n+1)+1,
                    \label{HS}
\end{align}
where the second row of inequalities follows by the maximum entropy
bound from \cite[Lemma 13.5.4]{CoverThomas06}. Hence, by the
inequality
\begin{align}
  \label{DPI}
  \sred\mathbb{H}(X|Y)\le \sred\mathbb{H}(X|f(Y))
\end{align}
 for a measurable function $f$, we obtain that
\begin{align}
  \sred\mathbb{H}(X_1^n)-\sred\mathbb{H}(X_1^n|Z_1^\infty)
  &\ge \sred\mathbb{H}(X_1^n|S_n)-\sred\mathbb{H}(X_1^n|Z_1^\infty,S_n)-\sred\mathbb{H}(S_n)
  \nonumber\\
  &\ge \sred\mathbb{H}(X_1^n|S_n)-\sred\mathbb{H}(X_1^n|Z_1^{S_n},S_n)-\sred\mathbb{H}(S_n)
  \nonumber\\
  &= \sred\mathbb{I}(X_1^n;Z_1^{S_n}|S_n)-\sred\mathbb{H}(S_n)
  \nonumber\\
  &\ge \sred\mathbb{H}(Z_1^{S_n}|S_n)-\sred\mathbb{H}(Z_1^{S_n}|X_1^n,S_n)-\sred\mathbb{H}(S_n)
  \nonumber\\
  &= \sred\mathbb{H}(Z_1^{S_n}|S_n)-\sred\mathbb{H}(S_n)
  \nonumber\\
  &\ge (\log 2)\sred S_n-2\sred\mathbb{H}(S_n)
  \nonumber\\  
  &\ge (\log 2)\sred S_n-2\kwad{\log(\sred S_n+1)+1}
    \label{IXZ}
    .
\end{align}
Now we observe that
\begin{align}
  \sred\mathbb{H}(X_1^n|Z_1^\infty)\ge
  \sred\mathbb{H}(X_1^n|X_{n+1}^\infty)=
  hn
\end{align}
since the sequence of random variables $Z_1^\infty$ is a measurable
function of the sequence of random variables $X_{n+1}^\infty$, as
shown in \cite{Debowski09,Debowski11b}. Hence we have
\begin{align}
  \sred\mathbb{H}(X_1^n)-\sred\mathbb{H}(X_1^n|Z_1^\infty)\le
  \sred\mathbb{H}(X_1^n)-hn.
  \label{IXZHh}
\end{align}
By inequalities (\ref{IXZ}) and (\ref{IXZHh}) and equality
(\ref{RedundancyMI}), we obtain inequality (\ref{FactsMI}).
\end{proof}

The algorithmic version of the theorem about facts and words follows
roughly the same idea, with some necessary adjustments.
\begin{theorem}[facts and mutual information II]
  \label{theoFactsMIAlg}
  Let $(X_i)_{i=1}^\infty$ be a stationary process over a finite
  alphabet.  We have inequality
  \begin{align}
    \hilberg_{n\rightarrow\infty} \sred \card U_a(X_1^n) &\le
    \hilberg_{n\rightarrow\infty} \sred\mathbb{I}_a(X_1^n;X_{n+1}^{2n})
                                                         .
    \label{FactsMIAlg}
 \end{align}
\end{theorem}
\begin{proof}
  Let us write $S_n:=\card U_a(X_1^n)$. Observe that
  \begin{align}
    \mathbb{H}_a(z_1^{S_n}|S_n)&\pge \mathbb{H}_a(z_1^{S_n})-\mathbb{H}_a(S_n)
                        \nonumber\\  
                      &\peq (\log 2) S_n-C-\mathbb{H}_a(S_n),
                        \label{HZSAlg}
    \\
  \mathbb{H}_a(S_n)&\ple 2\log( S_n+1),
                    \label{HSAlg}
    \\
  \mathbb{H}_a(K(z_1^{S_n}))&\ple 2\log( K(z_1^{S_n})+1)
  \nonumber\\
                      &\ple 2\log( S_n+1),
                    \label{HKZSAlg}
\end{align}
where the first row of inequalities follows by the algorithmic
randomness of $z_1^\infty$, whereas the second and the third row of
inequalities follow by the bounds $K(n)\ple 2\log_2 (n+1)$ for
$n\ge 0$ and $K(z_1^k)\ple 2k$. Moreover, for any a computable
function $f$ there exists a constant $C_f\ge 0$ such that
\begin{align}
  \label{DPIAlg}
  \mathbb{H}_a(x|y)\ple \mathbb{H}_a(x|f(y))+C_f
  .
\end{align}
Hence, we obtain that
\begin{align}
  \mathbb{H}_a(X_1^n)-\mathbb{H}_a(X_1^n|z_1^\infty)
  &\pge \mathbb{H}_a(X_1^n|S_n)-\mathbb{H}_a(X_1^n|z_1^\infty,S_n)-\mathbb{H}_a(S_n)
  \nonumber\\
  &\pge \mathbb{H}_a(X_1^n|S_n)-\mathbb{H}_a(X_1^n|z_1^{S_n},S_n)-\mathbb{H}_a(S_n)
  \nonumber\\
  &\pge \mathbb{I}_a(X_1^n;z_1^{S_n}|S_n)-\mathbb{H}_a(K(z_1^{S_n}))-\mathbb{H}_a(S_n)
  \nonumber\\
  &\pge \mathbb{H}_a(z_1^{S_n}|S_n)-\mathbb{H}_a(z_1^{S_n}|X_1^n,K(X_1^n),S_n)
    \nonumber\\
  &\qquad -\mathbb{H}_a(K(z_1^{S_n}))-\mathbb{H}_a(S_n)
  \nonumber\\
  &\pge \mathbb{H}_a(z_1^{S_n}|S_n)-C_g-\mathbb{H}_a(K(z_1^{S_n}))-\mathbb{H}_a(S_n)
  \nonumber\\
  &\pge (\log 2) S_n-C-C_g-\mathbb{H}_a(K(z_1^{S_n}))-2\mathbb{H}_a(S_n)
  \nonumber\\  
  &\pge (\log 2) S_n-6\log( S_n+1)-C-C_g
    .
\end{align}
Since $-\sred\log( S_n+1)\ge -\log(\sred S_n+1)$ by the Jensen
inequality then
\begin{align}
  \sred\mathbb{H}_a(X_1^n)-\sred\mathbb{H}_a(X_1^n|z_1^\infty)\pge
  (\log 2)\sred S_n-6\log(\sred S_n+1)-C-C_g
    \label{IXZAlg}
    .
\end{align}
Now we observe that
\begin{align}
  \sred\mathbb{H}_a(X_1^n|z_1^\infty)\ge \sred\mathbb{H}(X_1^n)\ge hn
\end{align}
since the conditional prefix-free Kolmogorov complexity with the
second argument fixed is the length of a prefix-free code. Hence we
have
\begin{align}
  \sred\mathbb{H}_a(X_1^n)-\sred\mathbb{H}_a(X_1^n|z_1^\infty)\le
  \sred\mathbb{H}_a(X_1^n)-hn.
  \label{IXZHhAlg}
\end{align}
By inequalities (\ref{IXZAlg}) and (\ref{IXZHhAlg}) and equality
(\ref{RedundancyMIAlg}), we obtain inequality (\ref{FactsMIAlg}).
\end{proof}
  
\section{Mutual information and PPM words}
\label{secMIWords}

In this appendix, we will investigate some algebraic properties of the
length of the PPM code to be used for proving the second part of the
theorem about facts and words.  First of all, it can be seen that
\begin{align}
  \label{PPMfirst}
  \mathbb{H}_{\PPM_k}(x_1^n)=
  \begin{cases}
    \displaystyle n\log D, & k=-1,
    \\
    \displaystyle k\log D
    +\sum_{u\in\mathbb{X}^k}\log\frac{(N(u|x_1^{n-1})+D-1)!}{(D-1)!\prod_{a=1}^D
      N(ua|x_1^n)!}, & k\ge 0.
  \end{cases}
\end{align}
Expression (\ref{PPMfirst}) can be further rewritten using notation
\begin{align}
  \log^* n&:=
  \begin{cases}
    0, & n=0,
    \\
    \log n!-n\log n+n, & n\ge 1,      
  \end{cases}
  \\
  \mathfrak{H}(n_1,...,n_l)&:=
                             \begin{cases}
                               \sum_{i=1:n_i>0}^l n_i
                               \log \okra{\frac{\sum_{j=1}^l n_j}{n_i}},
                               & \text{if $n_j>0$ exists},
                               \\
                               0, & \text{else},
                             \end{cases}
  \\
  \mathfrak{K}(n_1,...,n_l)&:=
                             \sum_{i=1}^l \log^* n_i-\log^*\okra{\sum_{i=1}^l
                             n_i}.
\end{align}
Then, for $k\ge 0$, we define
\begin{align}
  \mathbb{H}_{\PPM^0_k}(x_1^n)&:=
    \sum_{u\in\mathbb{X}^k}
    \mathfrak{H}\okra{
    {N(u1|x_1^{n})}
    ,...,
    {N(uD|x_1^{n})}
    }
    ,
    \\
  \mathbb{H}_{\PPM^1_k}(x_1^n)&:=
    \sum_{u\in\mathbb{X}^k}
    \mathfrak{H}\okra{
      {N(u|x_1^{n-1})}
      ,
      {D-1}
    }
  \nonumber\\
  &\quad
   -\sum_{u\in\mathbb{X}^k} \mathfrak{K}\okra{
    {N(u1|x_1^{n})}
    ,...,
    {N(uD|x_1^{n})}
    ,
    {D-1}
    }
    .
\end{align}
As a result for $k\ge 0$ we obtain
\begin{align}
  \label{PPMexact}
  \mathbb{H}_{\PPM_k}(x_1^n)&=k\log D+ \mathbb{H}_{\PPM^0_k}(x_1^n)+
                              \mathbb{H}_{\PPM^1_k}(x_1^n).
\end{align}
In the following, we will analyze the terms on the right-hand side of
(\ref{PPMexact}).

\begin{theorem}
  \label{theoPPMbound}
  For $k\ge 0$ and $n\ge 1$, we have
  \begin{align}
    \label{PPMboundI}
    \tilde D\card V(k|x_1^{n-1})
    &\le \mathbb{H}_{\PPM^1_k}(x_1^n)< D\card V(k|x_1^{n-1})\okra{2+\log n}.
  \end{align}
  where $\tilde D:=-D\log\okra{D^{-1}}!>0$.
\end{theorem}
\begin{proof}
  Observe that $\mathfrak{H}(0,D-1)=\mathfrak{K}(0,...,0,D-1)=0$.
  Hence the summation in $\mathbb{H}_{\PPM^1_k}(x_1^n)$ can be restricted to
  $u\in\mathbb{X}^k$ such that $N(u|x_1^{n-1})\ge 1$. Consider such a
  $u$ and write $N=N(u|x_1^{n-1})$ and $N_a=N(ua|x_1^n)$.

  Since $\mathfrak{H}(n_1,...,n_l)\ge 0$ and
  $\mathfrak{K}(n_1,...,n_l)\ge 0$ (the second inequality follows by
  subadditivity of $\log^* n$), we obtain first
  \begin{align}
    &\mathfrak{H}\okra{{N},{D-1}}-
      \mathfrak{K}\okra{{N_1},...,{N_D},{D-1}}
    \nonumber\\
    &\quad\le  \mathfrak{H}\okra{{N},{D-1}}
    \nonumber\\
    &\quad= N\log\okra{1+\frac{D-1}{N}}+(D-1)\log\okra{1+\frac{N}{D-1}}
    \nonumber\\
    &\quad\le N\cdot\frac{D-1}{N}+(D-1)\log\okra{1+\frac{N}{D-1}}
    \nonumber\\
    &\quad= (D-1)\kwad{1+\log\okra{1+\frac{N}{D-1}}}
    < D\okra{2+\log n}
    ,
    \label{PPMboundUpper}
  \end{align}
  where we use $\log(1+x)\le x$ and $N< n$.  On the other hand,
  function $\log^* n$ is concave so by $\sum_{a=1}^D N_a=N$ and the
  Jensen inequality for $\log^* n$ we obtain
\begin{align}
  &\mathfrak{H}\okra{{N},{D-1}}-\mathfrak{K}\okra{{N_1},...,{N_D},{D-1}}
  \nonumber\\
  &\quad\ge \mathfrak{F}\okra{{N},{D}}:=
    N\log\okra{1+\frac{D-1}{N}}+(D-1)\log\okra{1+\frac{N}{D-1}}
  \nonumber\\
  &\qquad\qquad\qquad
    +\log^*(N+D-1)-\log^*(D-1)-D\log^*\okra{N/D}
  \nonumber\\
  &\quad=
    \log(N+D-1)!-\log(D-1)!-D\log\okra{N/D}!-N\log D
  \nonumber\\
  &\quad=
    \log\frac{(N+D-1)!}{(D-1)!\okra{N/D}!^D D^N}
    \ge 0
\end{align}
since
\begin{align}
  \okra{N/D}!^D D^N
  &=N^D(N-D)^D(N-2D)^D\cdot...\cdot D^D
  \nonumber\\
  &\le
  (N+D-1)(N+D-2)\cdot...\cdot D
  =\frac{(N+D-1)!}{(D-1)!}
  .
\end{align}
Moreover, function $\mathfrak{F}\okra{{N},{D}}$ is growing in argument
$N$. Hence
\begin{align}
  \mathfrak{F}\okra{{N},{D}}\ge \mathfrak{F}\okra{{1},{D}}=
  -D\log\okra{D^{-1}}!
  .
  \label{PPMboundLower}
\end{align}
Summing inequalities (\ref{PPMboundUpper}) and (\ref{PPMboundLower})
over $u\in\mathbb{X}^k$ such that $N(u|x_1^{n})\ge 1$, we obtain the
claim.
\end{proof}

The mutual information is defined as a difference of entropies.
Replacing the entropy with an arbitrary function $\mathbb{H}_Q(u)$, we
obtain this quantity:
\begin{definition}
  The \emph{$Q$ pointwise mutual information} is defined as
  \begin{align}
    \mathbb{I}_Q(u;v):=
    \mathbb{H}_Q(u)+\mathbb{H}_Q(v)
    -\mathbb{H}_Q(uv)
    .
  \end{align}
\end{definition}
We will show that the $\PPM^0_k$ pointwise mutual information cannot
be positive.
\begin{theorem}
  \label{theoMathFrakH}
  For $n_i=\sum_{j=1}^l n_{ij}$, where $n_{ij}\ge 0$, we have
  \begin{align}
    \mathfrak{H}(n_1,...,n_k)\ge \sum_{j=1}^l
    \mathfrak{H}(n_{1j},...,n_{kj}).
  \end{align}
\end{theorem}
\begin{proof}
  Write $N:=\sum_{i=1}^k \sum_{j=1}^l n_{ij}$, $p_{ij}:=n_{ij}/N$,
  $q_i:=\sum_{j=1}^l p_{ij}$, and $r_j:=\sum_{i=1}^k p_{ij}$. We observe
  that
  \begin{align}
    \mathfrak{H}(n_1,...,n_k)- \sum_{j=1}^l
    \mathfrak{H}(n_{1j},...,n_{kj})
    =
    N\sum_{i=1}^k \sum_{j=1}^l p_{ij}\log\frac{p_{ij}}{q_i r_j},
  \end{align}
  which is $N$ times the Kullback-Leibler divergence between
  distributions $\klam{p_{ij}}$ and $\klam{q_i r_j}$ and thus is
  nonnegative.
\end{proof}
\begin{theorem}
  \label{theoPPMZeroMI}
  For $k\ge 0$, we have
  \begin{align}
    \label{PPMZeroMI}
    \mathbb{I}_{\PPM^0_k}(x_1^n;x_{n+1}^{n+m})\le 0.
  \end{align}
\end{theorem}
\begin{proof}
  Consider $k\ge 0$.  For $u\in\mathbb{X}^k$ and $a\in\mathbb{X}$, we
  have
  \begin{align}
    N(ua|x_1^{n+m})&=N(ua|x_1^n)+N(ua|x_{n-k}^{n+k})+N(ua|x_{n+1}^{n+m})
    .
  \end{align}
  Thus using Theorem \ref{theoMathFrakH} we obtain
  \begin{align}
    \mathfrak{H}\okra{
      {N(u1|x_1^{n+m})}
      ,...,
      {N(uD|x_1^{n+m})}
      }
    &\ge
    \mathfrak{H}\okra{
      {N(u1|x_1^n)}
      ,...,
      {N(uD|x_1^n)}
    }
    \nonumber\\
    &+
    \mathfrak{H}\okra{
      {N(u1|x_{n-k}^{n+k})}
      ,...,
      {N(uD|x_{n-k}^{n+k})}
    }
    \nonumber\\
    &+
    \mathfrak{H}\okra{
      {N(u1|x_{n+1}^{n+m})}
      ,...,
      {N(uD|x_{n+1}^{n+m})}
    }
    .
  \end{align}
  Since the second term on the right hand side is greater than or
  equal zero, we may omit it and summing the remaining terms over all
  $u\in\mathbb{X}^k$ we obtain the claim.
\end{proof}

Now we will show that the PPM pointwise mutual information between two
parts of a string is roughly bounded above by the cardinality of the
PPM vocabulary of the string multiplied by the logarithm of the string
length. 
\begin{theorem}
  \label{theoPPMVocabMI}
  We have
  \begin{align}
    \mathbb{I}_{\PPM}(x_1^n;x_{n+1}^{n+m})&\le 1
                                            +4\log
                                            \kwad{G_{\PPM}(x_1^{n+m})+2}
                                            +\kwad{G_{\PPM}(x_1^{n+m})+1}\log D
    \nonumber\\    
    &\qquad
    +
    2D \card V_{\PPM}(x_1^{n+m})\kwad{2+\log (n+m)}.
    \label{PPMVocabMI}
  \end{align}
\end{theorem}
\begin{proof}
  Consider $k\ge 0$. By Theorems \ref{theoPPMbound} and
  \ref{theoPPMZeroMI} we obtain
  \begin{align}
    \mathbb{I}_{\PPM_k}(x_1^n;x_{n+1}^{n+m})
    &=
    k\log D
      +\mathbb{I}_{\PPM^0_k}(x_1^n;x_{n+1}^{n+m})
      +\mathbb{I}_{\PPM^1_k}(x_1^n;x_{n+1}^{n+m})
    \nonumber\\    
    &\le
    k\log D
    +
    D\card V(k|x_1^{n})\kwad{2+\log n}
    \nonumber\\    
    &\quad\qquad
    +
    D\card V(k|x_{n+1}^{n+m})\kwad{2+\log m}
    \nonumber\\    
    &\le k\log D+2D\card V(k|x_1^{n+m})\kwad{2+\log(n+m)}
    .   
  \end{align}
  In contrast, $\mathbb{I}_{\PPM_{-1}}(x_1^n;x_{n+1}^{n+m})=0$.
  Now let $G=G_{\PPM}(x_1^{n+m})$. Since
  \begin{align}
     \mathbb{H}_{\PPM}(x_1^{n+m})\ge \mathbb{H}_{\PPM_G}(x_1^{n+m})
  \end{align}
  and
  \begin{align}
     \mathbb{H}_{\PPM}(u)\le \mathbb{H}_{\PPM_k}(u)+ 1/2+2\log (k+2)
  \end{align}
  for any $u\in\mathbb{X}^*$ and $k\ge -1$,
  we obtain
  \begin{align}
    \mathbb{I}_{\PPM}(x_1^n;x_{n+1}^{n+m})
    &\le
      \mathbb{I}_{\PPM_G}(x_1^n;x_{n+1}^{n+m})+1+4\log (G+2)
    \nonumber\\
    &\le 1+4\log (G+2)+(G+1)\log D
    \nonumber\\
    &\qquad 
    +2D\card V(G|x_1^{n+m})\kwad{2+\log(n+m)}
    .
  \end{align}
  Hence the claim follows.
\end{proof}

Consequently, we may prove the second part of Theorems
\ref{theoFactsWords} and \ref{theoFactsWordsAlg}, i.e., the theorems
about facts and words.
\begin{theorem}[mutual information and words]
  \label{theoMIWords}
  Let $(X_i)_{i=1}^\infty$ be a stationary process over a finite
  alphabet.  We have inequalities
  \begin{align}
    \hilberg_{n\rightarrow\infty} \sred\mathbb{I}(X_1^n;X_{n+1}^{2n})
    &\le \hilberg_{n\rightarrow\infty} \sred\mathbb{I}_a(X_1^n;X_{n+1}^{2n})
    \nonumber\\
    &\le \hilberg_{n\rightarrow\infty} \sred
    \kwad{G_{\PPM}(X_1^n)+\card V_{\PPM}(X_1^n)} .
    \label{MIWords}
 \end{align}
\end{theorem}
\begin{proof}
By Theorem \ref{theoPPMVocabMI}, we obtain
\begin{align}
  \hilberg_{n\rightarrow\infty} \sred \mathbb{I}_{\PPM}(X_1^n;X_{n+1}^{2n})
  &\le \hilberg_{n\rightarrow\infty} \sred
  \kwad{G_{\PPM}(X_1^n)+\card V_{\PPM}(X_1^n)} .
  \label{PPMMIWords}
\end{align}
In contrast, Theorems \ref{theoPPMUniversal} and
\ref{theoHilbergRedundancy} and inequalities (\ref{SourceCoding})
and (\ref{ShannonFanoPPM}) yield
\begin{align}
  \hilberg_{n\rightarrow\infty}
  \kwad{\sred\mathbb{H}(X_1^n)-hn}
  &\le
  \hilberg_{n\rightarrow\infty}
  \kwad{\sred\mathbb{H}_a(X_1^n)-hn}
  \nonumber\\
  &\le
  \hilberg_{n\rightarrow\infty}
  \kwad{\sred\mathbb{H}_{\PPM}(X_1^n)-hn}
  \nonumber\\
  &\le
  \hilberg_{n\rightarrow\infty} \sred \mathbb{I}_{\PPM}(X_1^n;X_{n+1}^{2n})
\end{align}
Hence by equalities (\ref{RedundancyMI}) and (\ref{RedundancyMIAlg}),
we obtain inequality (\ref{MIWords}).
\end{proof}

\section{Hilberg exponents for Santa Fe processes}
\label{secSantaFe}

We begin with a general observation for Hilberg exponents. In
\cite{Debowski15d} this result was discussed only for the Hilberg
exponent of mutual information.
\begin{theorem}[cf.\ \cite{Debowski15d}]
  \label{theoHilbergExp}
  For a sequence of random variables $Y_n\ge 0$, we have
\begin{align}
  \label{HilbergExp}
  \hilberg_{n\rightarrow\infty} Y_n \le
  \hilberg_{n\rightarrow\infty} \sred Y_n \text{ almost surely}
  .
\end{align}
\end{theorem}
\begin{proof}
Denote $\delta:=\hilberg_{n\rightarrow\infty} \sred Y_n$. From the
Markov inequality, we have
\begin{align}
  \sum_{k=1}^\infty
  P\okra{\frac{Y_{2^k}}{2^{k(\delta+\epsilon)}}\ge 1}
  &\le
  \sum_{k=1}^\infty
  \frac{\sred Y_{2^k}}{2^{k(\delta+\epsilon)}}
  \nonumber\\
  &\le
  A 
  +
  \sum_{k=1}^\infty
  \frac{2^{k(\delta+\epsilon/2)}}{2^{k(\delta+\epsilon)}}
  <\infty
  ,
\end{align}
where $A<\infty$. Hence, by the Borel-Cantelli lemma we have $Y_{2^k}<
2^{k(\delta+\epsilon)}$ for all but finitely many $n$ almost surely.
Since we can choose $\epsilon$ arbitrarily small, in particular we obtain
inequality (\ref{HilbergExp}).  
\end{proof}

In \cite{Debowski12} and \cite{Debowski15d} it was shown that the Santa Fe
process with exponent $\alpha$ satisfies equalities
\begin{align}
  \hilberg_{n\rightarrow\infty} \mathbb{I}(X_{-n+1}^0;X_1^{n}) 
  &=
    1/\alpha \text{ almost surely}
    ,
  \\
  \hilberg_{n\rightarrow\infty} \sred \mathbb{I}(X_{-n+1}^0;X_1^{n}) 
  &=
  1/\alpha
  .
\end{align}
We will now show a similar result for the number of probabilistic
facts inferrable from the Santa Fe process almost surely and in
expectation. Since Santa Fe processes are processes over an infinite
alphabet, we cannot apply the theorem about facts and words.

\begin{theorem}
  \label{theoFacts}
  For the Santa Fe process with exponent $\alpha$ we have
  \begin{align}
    \label{FactsAs}
  \hilberg_{n\rightarrow\infty} \card U(X_1^n) 
  &=
  1/\alpha \text{ almost surely}
  ,
    \\
    \label{FactsExp}
  \hilberg_{n\rightarrow\infty} \sred\card U(X_1^n) 
  &=
  1/\alpha
  .
\end{align}
\end{theorem}
\begin{proof}
First, we obtain
\begin{align}
  P(\card U(X_1^n)\le m_n)
  &\le
  \sum_{k=1}^{m_n}
  P(g(k;X_1^n)\neq Z_k)
  =
  \sum_{k=1}^{m_n}
  \kwad{1-P(K_i=k)}^n
  \nonumber\\
  &\le m_n\kwad{1-\frac{m_n^{-\alpha}}{\zeta(\alpha)}}^n
  \le m_n\exp\okra{-nm_n^{-\alpha}/\zeta(\alpha)},
\end{align}
where $\zeta(\alpha):=\sum_{k=1}^\infty k^{-\alpha}$ is the zeta
function. Put now $m_n=n^{1/\alpha-\epsilon}$ for an $\epsilon>0$. It
is easy to observe that
$\sum_{n=1}^\infty P(\card U(X_1^n)\le m_n)<\infty$. Hence by the
Borel-Cantelli lemma, we have inequality $\card U(X_1^n)> m_n$ for all
but finitely many $n$ almost surely.

Second, we obtain
\begin{align}
  P(\card U(X_1^n)\ge M_n)
  &\le
    \frac{n!}{(n-M_n)!}    
  \prod_{k=1}^{M_n}
    P(K_i=k)
  \nonumber\\
  &= \frac{n!}{(n-M_n)!(M_n!)^\alpha[\zeta(\alpha)]^{M_n}}.
\end{align}
Recalling from Appendix \ref{secMIWords} that
$\log n!=n(\log n -1)+\log^*n$, where $\log^*n\le \log(n+2)$ is
subadditive, we obtain
\begin{align}
  &\log P(\card U(X_1^n)\ge M_n)
  \nonumber\\
  &\qquad\le
  n(\log n -1)-(n-M_n)\kwad{\log(n-M_n)-1}
  \nonumber\\
  &\qquad\quad -\alpha M_n(\log M_n-1)
    +\log^* M_n-M_n\log \zeta(\alpha)
  \nonumber\\
  &\qquad\le  
    M_n\kwad{\log n -\alpha(\log M_n-1) -\log \zeta(\alpha)}
    +\log^* M_n
\end{align}
by $\log n\le \log(n-M_n)+\frac{M_n}{n}$. Put now $M_n=Cn^{1/\alpha}$
for a $C>e[\zeta(\alpha)]^{-1/\alpha}$.  We obtain
\begin{align}
  \label{StretchedExponential}
  P(\card U(X_1^n)\ge M_n)\le (Cn^{1/\alpha}+2)\exp(-\delta n^{1/\alpha})
\end{align}
where $\delta>0$ so
$\sum_{n=1}^\infty P(\card U(X_1^n)\ge M_n)<\infty$. Hence by the
Borel-Cantelli lemma, we have inequality $\card U(X_1^n)< M_n$ for all
but finitely many $n$ almost surely. Combining this result with the
previous result yields equality (\ref{FactsAs}).

To obtain equality (\ref{FactsExp}), we invoke Theorem
\ref{theoHilbergExp} for the lower bound, whereas for the upper bound
we observe that
\begin{align}
  \sred\card U(X_1^n)\le M_n+nP(\card U(X_1^n)\ge M_n)  
\end{align}
where the last term decays according to the stretched exponential bound
(\ref{StretchedExponential}) for $M_n=Cn^{1/\alpha}$.
\end{proof}

\bibliographystyle{IEEEtran}

\bibliography{0-journals-abbrv,0-publishers-abbrv,ai,mine,tcs,ql,books,nlp}

\end{document}